\documentclass[twocolumn, superscriptaddress, nofootinbib, floatfix, aps, pra, 10pt]{revtex4-2}

\usepackage{lmodern}
\usepackage{float}
\usepackage[colorlinks, allcolors=blue]{hyperref}
\usepackage[spalign, phfqit]{tchshortcuts}
\usepackage{siunitx}
\usepackage{amsthm}
\usepackage[normalem]{ulem}
\usepackage{graphicx}
\usepackage[capitalise,nameinlink]{cleveref}

\DeclareMathOperator{\maxCF}{maxCF}
\DeclareMathOperator{\QTCF}{QTCF}

\newtheorem{theorem}{Theorem}
\newtheorem{lem}{Lemma}
\newtheorem{definition}{Definition}
\newtheorem{remark}{Remark}
\newtheorem{prop}{Proposition}
\newtheorem{cor}{Corollary}

\theoremstyle{definition}
\newtheorem{eg}{Example}
\newtheorem*{eg*}{\Cref{eg:simplest_nontrivial} cont}

\theoremstyle{remark}

\usepackage{xcolor}
\newcommand\TC[1]{#1}
\newcommand\AC[1]{#1}

\newcommand{\oxfordaddress}{Department of Materials,
University of Oxford,
Parks Road,
Oxford OX1 3PH,
United Kingdom}

\begin{document}

\title{\TC{An Entanglement Monotone from} the Contextual Fraction}

\author{Tim Chan}
\email{timothy.chan@materials.ox.ac.uk}
\affiliation{\oxfordaddress}

\author{Andrei Constantin}
\email{andrei.constantin@physics.ox.ac.uk}
\affiliation{School of Mathematics, University of Birmingham, Watson Building, Birmingham B15 2TT, UK\\
Rudolf Peierls Centre for Theoretical Physics, University of Oxford, Parks Road, Oxford OX1 3PU, UK}

\begin{abstract}
The contextual fraction introduced by Abramsky and Brandenburger defines a quantitative measure of contextuality associated with empirical models,
i.e.\ tables of probabilities of measurement outcomes in experimental scenarios.
In this paper we define \TC{an entanglement monotone} relying on the contextual fraction.
We first show that any separable state is necessarily non-contextual \TC{with respect to any Bell} scenario.
Then, for \TC{2-qubit} states, we associate a \TC{state-dependent Bell scenario}
and show that the corresponding contextual fraction is \TC{an entanglement monotone},
suggesting contextuality may be regarded as a refinement of entanglement.
\TC{We~call this monotone the \emph{quarter-turn contextual fraction} (QTCF),
and use it to set an upper bound of approximately \num{0.601} for the minimum entanglement entropy needed to guarantee contextuality with respect to some Bell scenario.}
\end{abstract}

\maketitle


\section{Introduction}\label{Introduction}
Multiple quantum computing algorithms have been shown to significantly outperform their fastest-known classical counterparts~\cite{Shor1994,Shor_1997,Grover1996, Childs2002, Kempe2005, Buluta2009, Harrow2009, Brown2010,Georgescu2014,Farhi2014}. While it is difficult to systematically identify the resources responsible for quantum advantage, all the proposed algorithms leverage certain fundamental quantum phenomena such as superposition, entanglement, and contextuality, to achieve their computational advantage \cite{Jozsa2003, Curty2004, Howard2014,Raussendorf2013}. Understanding how these quantum phenomena relate to one another may provide valuable hints towards more effective and systematic strategies for algorithm design. 

The relation between entanglement and superposition is immediate: entanglement implies the presence of superposition. When quantum systems become entangled, their combined state cannot be expressed as a single product of states corresponding to each system. As such, entanglement provides a basis-independent characterisation of superposition for composite systems.

The relation between entanglement and contextuality is less obvious. Contextuality can be discussed both for single and composite quantum systems describing the dependence of measurement outcomes on the context in which they are measured, i.e.\ on the set of compatible measurements performed simultaneously. 
Various mathematical approaches to contextuality have been proposed (see Ref.~\cite{Budroni2021} for a recent review), ranging from the Kochen--Specker theorem~\cite{kochen_specker_1967}, to operational approaches~\cite{spekkens_2005}, sheaf-theoretic approaches~\cite{Abramsky2011}, the framework of contextuality-by-default~\cite{Dzhafarov2016_2}, graph and hypergraph-theoretic approaches~\cite{Cabello2014,Acin2015}, etc. While substantial overlap between the various notions of contextuality exist, the terminology and mathematical structures in each approach can be very different. Here, we focus on the sheaf-theoretic approach proposed by Abramsky and Brandenburger in Ref.~\cite{Abramsky2011}. A recent discussion of the relation between entanglement and the preparation-and-measurement-contextuality proposed by Spekkens~\cite{spekkens_2005} was given in Ref.~\cite{Plavala2022},
and between entanglement and the Kochen--Specker theorem in Ref.~\cite{wright2023contextuality}.

In the sheaf-theoretic approach, one views the measurement outcomes for a quantum system as local assignments of values defined on specific maximal sets of compatible measurements (maximal contexts). 
\AC{The central idea of this approach (as reviewed, e.g., in Ref.~\cite{Constantin2015}) is to use the standard mathematical framework of sheaf theory to model how such local measurement outcomes can or cannot be consistently glued together into a global assignment, thereby characterizing contextuality as the obstruction to such a global section.}
This idea was inspired by the work of Butterfield and Isham who showed that the Kochen--Specker theorem can be reformulated as the non-existence of global sections of a certain presheaf~\cite{ButterfieldIsham1998}. In this framework, the connection between nonlocality and contextuality becomes manifest, nonlocality being expressed as a particular case of contextuality~\cite{Fine1982, Abramsky2011}.

\AC{As discussed in this paper,} the connection between
superposition,
entanglement
and sheaf-theoretic contextuality \TC{with respect to local measurements for} composite states
is given by the chain of implications:
\begin{equation*}
\text{contextual}\Rightarrow\text{entangled}\Rightarrow\text{superposition},
\end{equation*}
which we explain in \cref{sec:CimpliesE}.
Furthermore, to quantify the relation between entanglement and contextuality
we use the standard entanglement entropy and the contextual fraction of Refs.~\cite{Abramsky2011, Abramsky2017}.

The role played by entanglement and contextuality in quantum computation has been recognised for a long time. Entanglement is a key ingredient in teleportation~\cite{Bennett1993} and cryptography~\cite{Ekert1991} protocols, and is necessary for quantum speedup~\cite{Jozsa2003}.
However, Van den Nest showed that a large quantity of entanglement is not necessary, suggesting perhaps  that entanglement alone is not the source of this speedup \cite{VandenNest2013}. In contrast, contextuality was proven to be necessary for quantum speedup in two computational paradigms \cite{Howard2014,BermejoVega2017}, the effectiveness of one (measurement-based quantum computation) even depending directly on the amount of contextuality~\cite{Abramsky2017}. 

The paper is structured as follows. After reviewing the necessary background in \cref{sec:background}, in \cref{sec:CimpliesE} we show that for arbitrary composite systems, sheaf-theoretic contextuality \TC{of local measurements} implies entanglement.
In \cref{sec:methods}, we focus on \TC{2-qubit} systems
and show that the contextual fraction computed with respect to a measurement scenario uniquely determined by the state
\TC{is an entanglement monotone.
We call this function the \emph{quarter-turn contextuality fraction} (QTCF).
Lastly,
we give an upper bound on the minimum entanglement entropy that a 2-qubit state must have in order to be contextual with respect to some Bell scenario.}

\section{Background}\label{sec:background}
In this section, we set the notation and introduce the basic concepts needed in the following sections. \AC{The entire discussion refers to pure states.}

\vspace{8pt}
\noindent{\bfseries Single Qubits.} We specify qubits by the standard Bloch sphere angles:
\begin{equation}\label{eq:Bloch}
\ket{\theta, \phi} :=
\cos\tfrac\theta2 \ket{0} +\e^{\i \phi}
\sin\tfrac\theta2 \ket{1}
\end{equation}
where
$\theta \in [0, \pi]$,
$\phi \in [0, 2 \pi)$,
and $`{\ket{0}, \ket{1}}$ denote, as usual, an orthonormal basis (ONB). In particular, any antipodal vectors
$`{\ket{\theta, \phi}, \ket{\pi -\theta, \pi +\phi}}$
form an ONB denoted by $\mathcal B(\theta, \phi)$. Moreover, any single-qubit ONB arises in this way. 
In this notation, the Pauli $z$ basis is $\mathcal B_z :=\mathcal B(0,0)=`{\ket{0}, \ket{1}}$,
while the Pauli $x$ and $y$ bases are
$\mathcal B_x :=\mathcal B(\pi/2,0) =`{\ket{0_x}, \ket{1_x}}$ and
$\mathcal B_y :=\mathcal B(\pi/2, \pi/2) =`{\ket{0_y}, \ket{1_y}}$,
respectively.

\begin{eg}
For later reference, the Pauli $x$ and $y$ bases rotated about the $z$-axis by $\pi/8$ will be denoted by
\begin{align}
\mathcal B_{ \pi/8} &:=\mathcal B`\Big(\frac \pi2, \frac  \pi 8), &
\mathcal B_{5\pi/8} &:=\mathcal B`\Big(\frac \pi2, \frac{5\pi}8).
\end{align}
\end{eg}
 
\vspace{8pt}
\noindent{\bfseries Multiple Qubits.} 
Systems of $n$ qubits correspond to vectors in $\mathcal H^{\otimes n}$, where $\mathcal H\simeq \mathbb C^2$ is the single-qubit Hilbert space. Any 2-qubit state $\ket{\psi} \in \mathcal H_\A \otimes \mathcal H_\B$ can be written as $\alpha \ket{00} +\beta \ket{11}+\gamma\ket{01}+\delta\ket{10}$.
A \emph{diagonal} state is of the form $\alpha \ket{00} +\beta \ket{11}$
and can be parametrised in a Bloch-sphere style as
\begin{equation}\label{eq:diag_state}
\ket{\diag;\theta,\phi} :=
\cos\tfrac\theta2 \ket{00} + \e^{\i \phi}
\sin\tfrac\theta2 \ket{11}.
\end{equation}
If $\phi=0$, the state is called \emph{real}.
The importance of real diagonal states comes from the following proposition which will be needed in \cref{sec:methods}.
\begin{prop}\label{prop:2Q_Schmidt}
Any 2-qubit state can be expressed as a real diagonal state rotated by local unitaries:
\begin{equation}\label{eq:2Q_Schmidt_decomposition}
\ket{\psi} =\hat U_\A \otimes \hat U_\B \ket{\diag;\theta,0},
\end{equation}
with $\theta \in [0, \pi]$. 
\end{prop}
\noindent This is a particular case of the Schmidt decomposition,
where $`{\cos \tfrac \theta2, \sin \tfrac \theta2}$ are the Schmidt coefficients and $\hat U_\A `{\ket{0}, \ket{1}}$ and $\hat U_\B `{\ket{0}, \ket{1}}$, the Schmidt bases of $\ket{\psi}$.

\vspace{8pt}
\noindent{\bfseries Entanglement.}
\TC{An entanglement measure is usually defined as any function
that meets a certain list of conditions;
these can be found in, e.g., Ref.~\cite{Horodecki2000}.}
The unique entanglement measure for a
2-qubit state $\ket{\psi} \in \mathcal H_\A \otimes \mathcal H_\B$ is the entanglement entropy~\cite{Horodecki2000}, given by
\begin{align}
S_\ent (\ket{\psi}) &:=-\tr(\hat \rho_\A \lg_2 \hat \rho_\A), \label{eq:entanglement_entropy}\\
\hat \rho_\A &:=\sum_{\beta=0}^1 `<\beta \proj{\psi} \beta>, \label{eq:reduced_density}
\end{align}
where $`{\ket{\beta}}$ is any ONB for $\mathcal H_\B$ and hence $`<\beta|\psi> \in \mathcal H_\A$ is a partial inner product. $S_\ent (\ket{\psi})$ varies from 0 for a separable state to 1 for a maximally entangled state.
Intuitively, this formula indicates that the more information is lost when one qubit is ignored, the more entangled the qubits must have been. The entropy \cref{eq:entanglement_entropy} measures the information loss, and the operator \cref{eq:reduced_density} is the result of ignoring the second qubit.

\begin{eg}
A maximally entangled $n$-qubit state is $\ket{\GHZ_n} :=\tfrac1{\sqrt2}(\ket{0}^{\otimes n} +\ket{1}^{\otimes n})$, known as the GHZ state. Applying \cref{eq:entanglement_entropy}, it follows that  $S_\ent (\ket{\GHZ_2}) =1$.
\end{eg}

\TC{The definition of an {\bfseries entanglement monotone} is similar to that of an entanglement measure
but with a weaker set of conditions~\cite{Plenio2007,Horodecki2009}.
\begin{definition}\label{def:entanglement_monotone}
An entanglement monotone for pure states is any function $E(\ket{\psi})$ that:
\begin{enumerate}
\setlength{\itemsep}{0pt}
    \item is nonnegative,
    \item vanishes for separable states,
    \item does not increase under local operations and classical communication (LOCC).
\end{enumerate}
\end{definition}}

\vspace{8pt}
\noindent{\bfseries Sheaf-Theoretic Contextuality.}\label{Sheaf-Theoretic Contextuality}
We start with a basic review of the sheaf-theoretic framework introduced in Ref.~\cite{Abramsky2011}.
A measurement scenario (or `scenario' for short) is a triplet $(X, \mathcal M, O)$ where $X$ a finite set of measurements, $\mathcal M$ a set of contexts, and $O$ a finite set of outcomes for each measurement. Each context $C\in\mathcal M$ is a subset of $X$, containing compatible measurements.
For qubits, $O =`{0,1}$.

\begin{eg}\label{eg:simplest_nontrivial}
The simplest nontrivial scenario involves two parties, Alice and Bob, who each have one qubit. Alice measures her qubit in one of the two distinct ONBs, $a_1$ or $a_2$, while Bob measures his qubit in the ONB $b$. The total set of ONBs is $X =`{a_1, a_2, b}$.
The two possible contexts are $`{a_1, b}$ and $`{a_2, b}$ so $\mathcal M =`{`{a_1, b}, `{a_2, b}}$. Note that $a_1$ and $a_2$ cannot be in the same context because Alice cannot simultaneously measure in both ONBs. Finally, $O =`{0,1}$.
\end{eg}

Alice and Bob run statistics on their measurement outcomes and the results are stored in probability tables, referred to as empirical models (or `models' for short).
\Cref{tab:empirical_model_examples} gives three distinct models.

\begin{table}[H]
	\caption{Three empirical models associated with the scenario in \cref{eg:simplest_nontrivial}.
    \TC{Note these are for illustration purposes only,
    as (b) and (c) violate no-signalling.}}
	\centering
	\begin{tabular}{cc}
	(a) & (b)  \\
		$\begin{array}{|cc|cccc|}
		\hline
		\A & \B & 00 & 01 & 10 & 11 \\
		\hline
		a_1 & b & 1 & 0 & 0 & 0 \\
		a_2 & b & 1 & 0 & 0 & 0 \\
		\hline
		\end{array}$  &
		$\begin{array}{|cc|cccc|}
		\hline
		\A & \B & 00 & 01 & 10 & 11 \\
		\hline
		a_1 & b & 1 & 0 & 0 & 0 \\
		a_2 & b & 0 & 0 & 0 & 1 \\
		\hline
		\end{array}$ 
	\end{tabular}
		\begin{tabular}{ccc}
		 (c) \\
		$\begin{array}{|cc|cccc|}
		\hline
		\A & \B & 00 & 01 & 10 & 11 \\
		\hline
		a_1 & b & 1 & 0 & 0 & 0 \\
		a_2 & b & \frac12 & 0 & 0 & \frac12 \\
		\hline
		\end{array}$
	\end{tabular}
	\label{tab:empirical_model_examples}
\end{table}

\begin{definition} A model $e$ for a given scenario $(X,\mathcal M, O)$  is a set of probability distributions $e_C$, one for each context $C\in\mathcal M$, such that $e_C(\v s)$ is the probability to observe the joint outcome $\v s \in O^{C}$ in the context~$C$.
\end{definition}

\begin{remark}
A model $e$  corresponds to a probability table, $e_C$ denotes a probability distribution (a row in the probability table) and $e_C(\v s)$ a probability (an entry in the probability table).
\end{remark}

\begin{eg*}
Were it possible to perform \emph{all} measurements $`{a_1, a_2, b}$ at once, we would get a global assignment $\v g$, i.e.\ a joint outcome of \emph{all} measurements. For instance, a global assignment $\v g =010$ would imply that measurement $a_1$ yielded outcome 0, $a_2$ outcome 1, and $b$ outcome 0. Now consider the global probability distribution
\begin{equation}
d(\v g) =\begin{cases*}
1 &if $\v g =000$, \\
0 &otherwise,
\end{cases*}
\end{equation}
defined over the set $`{0,1}^3$ of global assignments.
Consider also the model of \cref{tab:empirical_model_examples}a.
Note how its probability distributions are marginals of~$d$. The fact that we can find such a global distribution is what makes \cref{tab:empirical_model_examples}a a noncontextual model. Such a global distribution does not exist for \cref{tab:empirical_model_examples}b nor \cref{tab:empirical_model_examples}c, hence they are both contextual models.
\end{eg*}

\begin{definition}
\noindent A model $e$ is noncontextual iff its probability distributions are marginals of a single distribution~$d$ over the set 
of global assignments (joint outcomes of all measurements in $X$), that is
	\begin{equation}\label{eq:NC_condition}
	\exists d: e_C =d|_C\quad \forall C\in \mathcal M,
	\end{equation}
	where $d|_C$ denotes $d$ marginalised to $C \subseteq X$. A model is contextual iff it is not noncontextual.
\end{definition}

This definition captures the intuition that probabilities of outcomes in a contextual model depend on the context, so cannot depend solely on pre-existing values. 
In the three-layer hierarchy of Refs.~\cite{Abramsky_2012, Abramsky_2014}, the above definition corresponds to the most basic notion of contextuality, called weak contextuality, the other two notions being logical and strong contextuality. 

\vspace{8pt}
\noindent{\bfseries The Contextual Fraction.} 
To measure contextuality, we consider the question: what fraction of a model $e$ is actually contextual?
The authors of Ref.~\cite{Abramsky2011} defined the contextual and noncontextual fractions of $e$ as
\begin{align}
\CF(e) &:=1 -\NCF(e), \label{eq:CF_definition} \\
\NCF(e) &:=\max`*{\lambda \in [0, 1]:e =\lambda e^\NC +(1-\lambda)e'}, \label{eq:NCF_definition}
\end{align}
where $\lambda$ is maximised by varying the models $e^\NC, e'$ with the constraint that $e^\NC$ be noncontextual.
This can be done with linear programming~ \cite{Abramsky2017}.
$\CF(e)$ measures how contextual $e$ is, varying from 0 for a noncontextual model to 1 for a strongly contextual model.
\begin{eg}
The maximum value of $\lambda$ for \cref{tab:empirical_model_examples}b is~0, corresponding to a contextual fraction of 1, so this is a strongly contextual model. \Cref{tab:empirical_model_examples}c is half of \cref{tab:empirical_model_examples}a plus half of \cref{tab:empirical_model_examples}b, and it turns out that half of \cref{tab:empirical_model_examples}a is the maximum noncontextual contribution one can find for \cref{tab:empirical_model_examples}c, so its contextual fraction is~$1/2$.
\end{eg}

\vspace{8pt}
\noindent{\bfseries Bell Scenarios.}\label{The Bell Scenario}
\TC{A scenario involving $n$ parties,
each possessing one part of an $n$-partite state
and choosing to measure their part in one of $k$ ONBs,
each choice with $l$ possible outcomes,
is called an $(n, k, l)$ Bell scenario \cite[Example 2.2.1]{Constantin2015}.
For instance,
Alice chooses from $`{a_1, a_2, \dots, a_k}$;
Bob, $`{b_1, \dots, b_k}$;
Charlie, $`{c_1, \dots, c_k}$;
etc.
So, $n$ parties with $k$ choices each means there are $k^n$ contexts.
When all parties choose from the same set of ONBs and $k =2$,
the scenario is completely specified by Alice's set $`{a_1, a_2}$
and we denote it as $\mathcal S(a_1, a_2)$.
All scenarios in this paper are Bell scenarios.
}
\begin{eg}\label{eg:max_measurements_GHZ}
$\mathcal S(\mathcal B_x, \mathcal B_y)$ is the scenario of the Pauli $x$ and $y$ bases. $\mathcal S`(\mathcal B_{\pi/8}, \mathcal B_{5 \pi/8})$ is the scenario of the Pauli $x$ and $y$ bases rotated about the $z$-axis by $\pi/8$.
\end{eg}

\noindent Given an \TC{$(n, k, l)$ Bell scenario
and an $n$-partite state $\ket{\psi}$,
we can construct a model $e$ by computing each probability $e_C(\v s), \v s \in `{0, \dots, l -1}^n$
via Born's rule e.g.\ for $n =2$:
$e_{a_1 b_2}(34) =|\braket{3_1 4_2}{\psi}|^2$.
More generally,
\begin{align}
    e_{a_i b_j \dots}(\alpha \beta \dots)
    &=|\bra{\alpha_i} \bra{\beta_j} \dots\ket{\psi}|^2 \label{eq:Born} \\
    \text{where } i, j, \dots &\in `{1, \dots, k} \notag\\
    \text{and } \alpha, \beta, \dots &\in `{0, \dots, l -1}. \notag
\end{align}}
We call $e$ the \emph{model of $\ket{\psi}$ with respect to the scenario} and refer to $\CF(e)$ as the \emph{contextual fraction of $\ket{\psi}$ with respect to the (measurements defined by the) scenario}.

\section{Contextuality implies entanglement}\label{sec:CimpliesE}

Given the above framework for contextuality, we begin the discussion on the relation between entanglement and contextuality with the following result which holds for $n$-\TC{partite states under local measurement on each party}.

\begin{lem}[Separable implies noncontextual]\label{lem:separable_implies_noncontextual}
A separable state is \TC{noncontextual with respect to any Bell} scenario.
\end{lem}
\begin{proof}
Let $\ket{\psi} =\ket{\psi_\A}\ket{\psi_\B} \dots$ be a separable state.
Substitute this into \cref{eq:Born}:
\TC{
\begin{align}
e_{a_i b_j \dots}(\alpha \beta \dots) &=|
`<\alpha_i|\psi_\A>
`<\beta_j |\psi_\B>
\dots|^2 \notag\\
&=\pr_{a_i}(\alpha) \pr_{b_j}(\beta) \dots \label{eq:product_PMF}
\end{align}
where $\pr_{a_i}(\alpha) =|`<\alpha_i|\psi_\A>|^2$ is the probability that Alice obtains outcome $\alpha$
when she chooses to measure in the ONB $a_i$.
Define a global probability distribution
\begin{align}
    d(\alpha_1 \alpha_2 \alpha_3 \dots \beta_1 \beta_2 \beta_3 \dots)
    &=`\Big[\prod_i \pr_{a_i}(\alpha_i)] \notag \\
    &\times `\Big[\prod_j \pr_{b_j}(\beta_j)] \times \dots \label{eq:product_d}
\end{align}
and consider the RHS of \cref{eq:NC_condition} for, say,
the context $C=`{a_1, b_1, \dots}$:
\begin{align}
d|_C (\alpha_1 \beta_1 \dots)
&\overset{\eqref{eq:product_d}}=\pr_{a_1}(\alpha_1) \pr_{b_1}(\beta_1) \dots \notag\\
&\overset{\eqref{eq:product_PMF}}=e_C(\alpha_1 \beta_1 \dots). \label{eq:separable_NC_condition}
\end{align}}
By similar calculations, \cref{eq:separable_NC_condition} holds for all other contexts so the noncontextuality condition \cref{eq:NC_condition} is satisfied.
\end{proof}
\begin{prop}[Contextual implies entangled]\label{prop:contextual_implies_entangled}
A~state which is contextual with respect to some \TC{Bell} scenario must be entangled.
\end{prop}
\begin{proof}
This is the contraposition of \cref{lem:separable_implies_noncontextual}.
\end{proof}

\Cref{prop:contextual_implies_entangled} indicates that contextuality gives a means for detecting entanglement.
If given a composite state \TC{$\ket{\psi}$} there exists a contextual model $e_{\psi}$
\TC{with respect to a Bell scenario},
then \TC{$\ket{\psi}$} must be entangled.
Likewise, as discussed in the Introduction, if \TC{$\ket{\psi}$} is entangled,
\TC{$\ket{\psi}$} is necessarily a superposition state, irrespective of the basis used.
Hence, for composite states \TC{and local measurements}
we have the following implications:
\begin{equation}
\text{contextual}\Rightarrow\text{entangled}\Rightarrow\text{superposition}.
\end{equation}
As such, if contextual resources are needed for some quantum computation, one must use entangled states. To make this more quantitative, we need to define a measure for the capacity of a state to be contextual. We do so below by defining a unique scenario associated to the state and evaluating the corresponding contextual fraction. 

\section{An Entanglement Monotone from Contextuality}\label{sec:methods}
\TC{
It would be nice to have a contextuality quantification
that can be directly compared with entanglement entropy.
However, as discussed in \cref{sec:background},
the contextual fraction
depends not only on the state but also on the measurement scenario,
and is computed from the solution to a non-trivial optimisation problem
for which no analytical expression is known.
Ideally,
one would assign to a state the following quantity.
\begin{definition}
    The $\maxCF$ of an $n$-qubit state $\ket{\psi}$ is
    \begin{equation}
        \maxCF(\ket{\psi}) :=\max`{\CF(e): e \in M}
    \end{equation}
    where $M$ the set of all models that can be produced by $\ket{\psi}$
    with respect to an $(n, 2, 2)$ Bell scenario.
\end{definition}
This corresponds to a double optimisation problem,
where the second optimisation refers to the optimisation over all possible Bell scenarios.
However,
this is too computationally intensive even for $n=2$,
so we instead begin by considering some examples.

\subsection{2-Qubit Examples}
\noindent
\textbf{Fixing the State, Varying the Scenario.}}
It is instructive to plot (see \cref{fig:CF_GHZ}) the contextual fraction of the $\ket{\GHZ_2}$ state with respect to the scenario
\begin{equation}\label{eq:equatorial_scenario}
\mathcal S`*(\mathcal B`\Big(\frac \pi2, \phi_1), \mathcal B`\Big(\frac \pi2, \phi_2)).
\end{equation}
\begin{figure}[H]
	\centering
	\includegraphics[width=.3\textwidth]{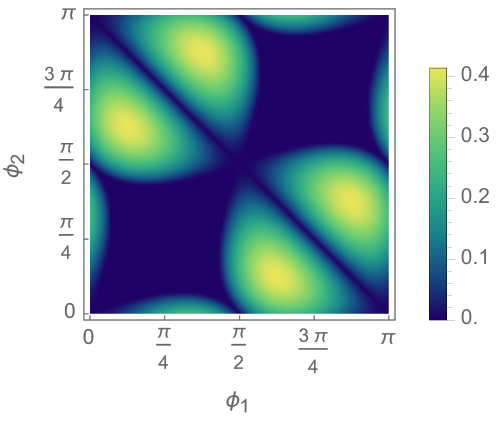}
	\caption{The contextual fraction (indicated by colour according to the colourbar on the right) of $\ket{\GHZ_2}$ with respect to \TC{the Bell scenario \cref{eq:equatorial_scenario} of} equatorial Bloch measurements.
    Each pair of azimuthal angles $(\phi_1, \phi_2)$ corresponds to a different empirical model.
	}
	\label{fig:CF_GHZ}
\end{figure}
\noindent
Note that $\ket{\GHZ_2}$ is not strongly contextual with respect to any choice of $(\phi_1, \phi_2)$.
Nonetheless, there are 4~maxima, one of which is at $(\phi_1, \phi_2) =(\pi/8, 5 \pi/8)$ corresponding to the scenario $\mathcal S`(\mathcal B_{\pi/8}, \mathcal B_{5 \pi/8})$. At these maxima, the contextual fraction is $\sqrt2 -1 \approx 0.414$.

\begin{remark}\label{rem:rotate_each_qubit}
Abramsky et al.~\cite[Proposition 4]{Abramsky2017} show that rotating each qubit about the $z$-axis by an angle $\varphi$ is equivalent to introducing a relative phase $2\varphi$ between the two terms in a diagonal state. More concretely, $(\phi_1, \phi_2) \to (\phi_1 +\varphi, \phi_2 +\varphi)$ in \cref{eq:equatorial_scenario} is equivalent to $\phi \to \phi +2 \varphi$ in \cref{eq:diag_state}. Thus the model of $\ket{\GHZ_2}$ with respect to $\mathcal S`(\mathcal B_{\pi/8}, \mathcal B_{5 \pi/8})$ is identical to that of $\ket{\diag; \pi/2, \pi/4}$ with respect to $\mathcal S(\mathcal B_x, \mathcal B_y)$.
\end{remark}

\noindent
\TC{\textbf{Fixing the Scenario, Varying the State.}}
As noted in \cref{prop:2Q_Schmidt}, any 2-qubit state can be rotated by local unitaries to a real diagonal state.
In this process, entanglement entropy is preserved.
\begin{lem}\label{lem:EE_general_2Q_state}
Local unitaries do not affect entanglement entropy
i.e.\ for any 2-qubit state $\ket{\psi}$:
\begin{align}
S_\ent(\ket{\psi})&= S_\ent(\ket{\diag; \theta, 0})\\
&= S_\ent(\ket{\diag; \theta, \phi})\quad\forall \phi \in \mathbb R,
\end{align}
where $\theta$ depends on $\ket{\psi}$ via \cref{eq:2Q_Schmidt_decomposition}.
\end{lem}
\begin{proof}
Compute the entanglement entropy of the diagonal state for arbitrary $\phi$ using the ONB $`{\ket{0}, \ket{1}}$:
\begin{align}
\hat \rho_\A &\overset{\eqref{eq:diag_state}}=\cos^2 \tfrac \theta2 \proj{0} +\sin^2 \tfrac \theta2 \proj{1} \\
\Rightarrow
S_\ent(\ket{\diag;\theta,\phi})
&=-2\big(\cos^2 \tfrac \theta2 \lg_2 \cos \tfrac \theta2 \notag \\
&\qquad +\sin^2 \tfrac \theta2 \lg_2 \sin \tfrac \theta2\big). \label{eq:EE_diag_state}
\end{align}
Compute also $S_\ent(\ket{\psi})$, this time using the ONB $\hat U_\B`{\ket{0}, \ket{1}}$, with $\hat U_\B$ defined by \cref{eq:2Q_Schmidt_decomposition}. This yields $\hat \rho_\A =\hat U_\A `\big(\cos^2 \tfrac \theta2 \proj{0} +\sin^2 \tfrac \theta2 \proj{1}) \hat U_\A^\dag$ and so $S_\ent(\ket{\psi})$ is identical to $S_\ent(\ket{\diag; \theta, \phi})$.
\end{proof}
 
\begin{figure}
	\centering
	\includegraphics[width=0.45\textwidth]{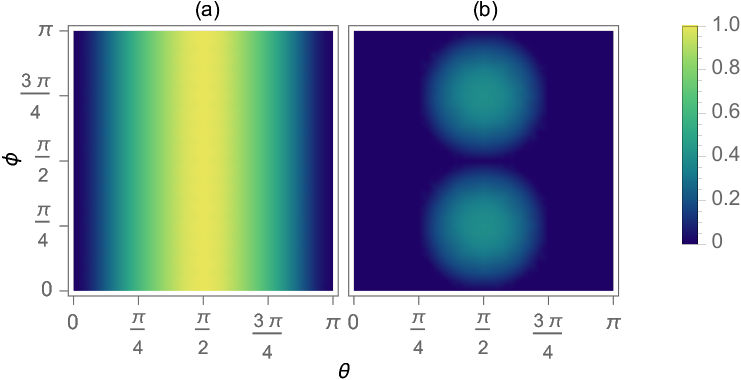}
	\caption{(a) Entanglement entropy of the diagonal state $\cos\tfrac\theta2 \ket{00} +\e^{\i \phi}\sin\tfrac\theta2 \ket{11}$. (b) The contextual fraction of said state with respect to \TC{the Bell scenario of} Pauli $x$ and $y$ basis measurements. The peak at $(\pi/2, \pi/4)$ has a contextual fraction of $\sqrt2 -1$ and corresponds to the empirical model described in \cref{rem:rotate_each_qubit}.}
	\label{fig:EE_CF_diag}
\end{figure}
\noindent
The entanglement entropy \cref{eq:EE_diag_state} of the diagonal state
is plotted in \cref{fig:EE_CF_diag}a as a function of $\theta$ and $\phi$,
while its contextual fraction with respect to $\mathcal S(\mathcal B_x, \mathcal B_y)$
is shown in \cref{fig:EE_CF_diag}b. We see a positive correlation: maxima of contextuality coincide with maxima of entanglement. 

\subsection{\TC{Introducing the QTCF}}

\TC{As previously mentioned,
computing $\maxCF$ is too computationally intensive.
We work around this difficulty by defining a unique Bell scenario from the state
and computing the corresponding contextual fraction.
Clearly,
we would like the scenario to try to maximise the contextuality.
The idea of the following scenario
is to map $\ket{\psi}$ to somewhere on the line $\phi =\pi/4$ in \cref{fig:EE_CF_diag},
as this was the angle (as well as $\phi =3 \pi/4$)
that maximised the contextual fraction for a given $\theta$ in a diagonal state.

\begin{definition}
    Let $\ket{\psi}$ be a 2-qubit state.
    Extract the local unitaries $\hat U_\A, \hat U_\B$ from $\ket{\psi}$
    via the Schmidt decomposition \cref{eq:2Q_Schmidt_decomposition}.
    The \emph{quarter-turn scenario} of $\ket{\psi}$
    is a $(2, 2, 2)$ Bell scenario $\mathcal S_\psi :=`{a_1, a_2} \times `{b_1, b_2}$ in which
    we set Alice's pair of ONBs as
    $`{\mathcal B_{\pi/8}, \mathcal B_{5 \pi/8}}$ rotated by $\hat U_\A$:
    \begin{equation}\label{eq:QTCF_A_basis}
    `{a_1, a_2} =\hat U_\A `{\mathcal B_{\pi/8}, \mathcal B_{5 \pi/8}},
    \end{equation} 
    and similarly for Bob:
    \begin{equation}\label{eq:QTCF_B_basis}
    `{b_1,b_2} =\hat U_\B `{\mathcal B_{\pi/8}, \mathcal B_{5 \pi/8}}.
    \end{equation}
\end{definition}

\noindent
We next prove that this indeed maps the space of all 2-qubit states
to the line $\phi =\pi/4$ in \cref{fig:EE_CF_diag}.

\begin{lem}\label{lem:model_identity}
    The model $e$ of a 2-qubit state $\ket{\psi}$ with respect to $\mathcal S_\psi$
    is identical to the model $f$ of $\ket{\diag; \theta, \pi/4}$
    with respect to $\mathcal S(\mathcal B_x, \mathcal B_y)$,
    where $\theta$ depends on $\ket{\psi}$ via \cref{eq:2Q_Schmidt_decomposition}.
\end{lem}

\begin{proof}
If e.g.\ from $\mathcal S_\psi$
Alice chooses $a_1$ and Bob chooses $b_2$, the 01 state is $\ket{0_1}\ket{1_2} =\hat U_\A \ket{0_{\pi/8}} \otimes \hat U_\B \ket{1_{5 \pi/8}}$. Substituting this and \cref{eq:2Q_Schmidt_decomposition} into \cref{eq:Born} and noting that $\hat U_\A, \hat U_\B$ cancel due to unitarity, we find
\begin{align}
e_{a_1 b_2}(01) &=\abs{`<0_{\pi/8} 1_{5 \pi/8}|{\diag}; \theta, 0>}^2 \notag\\
&\overset{\textnormal{\cref{rem:rotate_each_qubit}}}=
\abs{`<0_x 1_y|{\diag}; \theta, \pi/4>}^2 \notag\\
&=f_{a_1 b_2}(01). \label{eq:max_model_eg}
\end{align}
Similar calculations hold for all other probabilities in the model,
so $e \equiv f$.
\end{proof}

Having defined a state-dependent scenario,
we can define its corresponding contextual fraction.

\begin{definition}\label{def:QTCF}
    The QTCF of a 2-qubit state $\ket{\psi}$ is
    the contextual fraction of $\ket{\psi}$ with respect to
    its quarter-turn scenario $\mathcal S_\psi$.
\end{definition}

\noindent
In this way, we obtain a quantity that depends only on the state,
which allows us to compare with the entanglement entropy.
While we have not shown that $\QTCF =\maxCF$,
we nevertheless have an way for associating a definite measure of contextuality with any 2-qubit state.
Moreover, we show in the next subsection that QTCF is an entanglement monotone.
Combining \cref{lem:model_identity} with \cref{def:QTCF}
gives the following shortcut to computing QTCF in practice.

\begin{cor}\label{cor:QTCF_identity}
$\QTCF(\ket{\psi})$ equals
the contextual fraction of $\ket{\diag; \theta, \pi/4}$
with respect to $\mathcal S(\mathcal B_x, \mathcal B_y)$,
where $\theta$ depends on $\ket{\psi}$ via \cref{eq:2Q_Schmidt_decomposition}.
\end{cor}
}

\subsection{\TC{QTCF as an Entanglement Monotone}}
Equipped with \TC{QTCF},
we can show that an entangled state has a higher capacity to be contextual than a less entangled one:
\begin{prop}\label{prop:maxCF_is_monotone_WRT_EE}
\TC{QTCF} monotonically grows with entanglement entropy,
that is for any pair $\ket{\psi}, \ket{\chi}$ of 2-qubit states, $S_\ent(\ket{\psi}) >S_\ent(\ket{\chi})$ implies
\begin{equation}
\TC{\QTCF(\ket{\psi}) \ge\QTCF(\ket{\chi}).}
\end{equation}
\end{prop}
\begin{proof}
Consider any 2-qubit state $\ket{\psi}$ and its $\theta$ parameter defined by \cref{eq:2Q_Schmidt_decomposition}.
By \cref{lem:EE_general_2Q_state}, we have $S_\ent(\ket{\psi}) =S_\ent(\ket{\diag; \theta, \pi/4})$.
\TC{Given this and \cref{cor:QTCF_identity},
$\theta$ alone determines both entanglement entropy and QTCF.}
This dependence is plotted in \cref{fig:Monotone}a,
obtained as a cross-section of \cref{fig:EE_CF_diag}b at the line $\phi =\pi/4$.
By inverting $S_\ent(\ket{\psi})$ as a function of $\theta$,
we can find \TC{$\QTCF(\ket{\psi})$} as a function of $S_\ent(\ket{\psi})$,
which is plotted in \cref{fig:Monotone}b.
This function grows monotonically.
\begin{figure}
	\centering
	\includegraphics[width=0.45\textwidth]{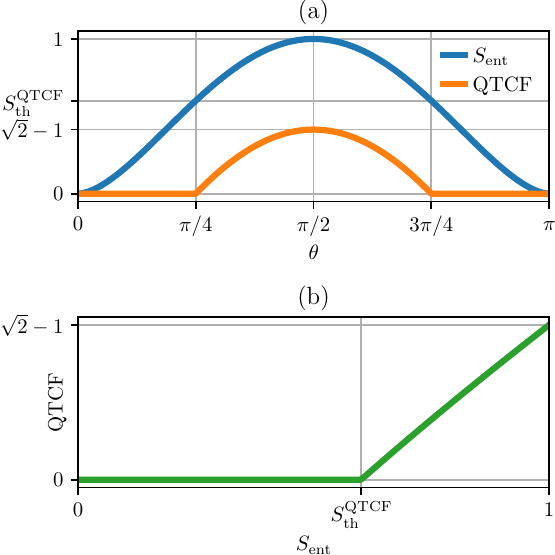}
	\caption{(a) Entanglement entropy and \TC{QTCF} as functions of the state's $\theta$-parameter defined by the Schmidt decomposition \cref{eq:2Q_Schmidt_decomposition}. (b) \TC{QTCF} as a function of entanglement entropy; $S_\th^{\QTCF}$ is the entanglement entropy for $\theta = \pi/4, 3 \pi/4$, given by \cref{eq:S_th_value}.}
	\label{fig:Monotone}
\end{figure}
\end{proof}

This analysis suggests that an increasing amount of contextuality requires an increasing amount of entanglement, at least in the case of two qubits.
\TC{
\noindent We now prove our main result.
\begin{theorem}
    QTCF is an entanglement monotone.
\end{theorem}
\begin{proof}
    We show that QTCF satisfies all three conditions of \cref{def:entanglement_monotone}.
    \begin{enumerate}
        \item The contextual fraction is in $[0, 1]$ so QTCF is nonnegative.
        \item By \cref{lem:separable_implies_noncontextual},
        any separable state has zero contextual fraction for all Bell scenarios,
        so its QTCF will vanish.
        \item Entanglement entropy is an entanglement monotone so does not increase under LOCC.
        So by \cref{prop:maxCF_is_monotone_WRT_EE},
        QTCF also cannot increase under LOCC. \qedhere
    \end{enumerate}
\end{proof}
}

\noindent
In addition to monotonicity,
\cref{fig:Monotone} suggests a threshold level of entanglement
required for \TC{nonzero QTCF:
let $A_{\QTCF}$ be the set of all 2-qubit states with $\QTCF({\ket{\psi}}) >0$ then}
\begin{align}
S_\th^{\QTCF} &\TC{:=\inf`{S_\ent(\ket{\psi}): \ket{\psi} \in A_{\QTCF}}} \notag\\
&=S_\ent`\Big(\ket[\big]{\diag; \frac \pi4, \phi}) \notag \\
&\overset{\eqref{eq:EE_diag_state}}=\frac1{2^2} `*[6 +\sqrt2 \lg_2 `\big(3-2 \sqrt2)] \approx 0.601. \label{eq:S_th_value}
\end{align}
\TC{Though this is not a threshold for
contextuality of the 2-qubit system across \emph{all} Bell scenarios,
we can use it as an upper bound:

\begin{prop}
    $S_\th^{\QTCF} \approx 0.601$ is an upper bound for the minimum entanglement entropy
    a 2-qubit state needs for contextuality with respect to at least one Bell scenario.
    More concretely,
    let $A_{\maxCF}$ be the set of all 2-qubit states $\ket{\psi}$ with $\maxCF(\ket{\psi}) >0$,
    and
    \begin{equation}
        S_\th^{\maxCF} :=\inf`{S_\ent(\ket{\psi}): \ket{\psi} \in A_{\maxCF}},
    \end{equation}
    then $S_\th^{\maxCF} \le S_\th^{\QTCF}$.
\end{prop}

\begin{proof}
    By their definitions,
    $\maxCF \ge \QTCF$ always,
    so $A_{\QTCF} \subseteq A_{\maxCF}$.
\end{proof}}

\section{Conclusion}\label{Conclusion}
We have compared entanglement entropy with the contextual fraction
and proved \TC{that multipartite states contextual under local measurements} must be entangled.
Since contextuality is necessary for quantum speedup the implication is that so is entanglement. This agrees with Jozsa \& Linden \cite{Jozsa2003}, who showed that some nonzero entanglement is necessary for exponential quantum speedup.

We have also developed \TC{the notion of quarter-turn contextual fraction (QTCF)} for any 2-qubit state
\TC{and showed that it is an entanglement monotone.}
This suggests that a large quantity of entanglement is necessary, e.g., in measurement-based quantum computing.
\TC{The QTCF} brings together two separate hierarchies, entanglement and contextuality, into a single picture.
\TC{Lastly,
we have found an upper bound for the minimum entanglement entropy a 2-qubit state needs
for contextuality with respect to some Bell scenario.
We do not know if this bound is sharp;
it would be interesting future work to try tightening this bound.
The quarter-turn scenario and thus the QTCF
generalise naturally to \num{>2}-qubit states
via higher-order singular value decomposition.
Another avenue of further investigation would be
to compare this generalised QTCF with \num{>2}-qubit entanglement measures.
}

\vspace{21pt}
{\bfseries Acknowledgements.} We are grateful to Samson Abramsky and Carmen Constantin for useful discussions. \AC{We also thank the anonymous referees for valuable comments and suggestions, which have contributed to clarifying and improving the manuscript.} TC would like to express his gratitude to the Oxford Physics Department for the award of the Gibbs Prize for the Best MPhys Research Project in 2022, on which this work is based. AC's research is supported by a Royal Society Dorothy Hodgkin Fellowship. 

\bibliography{tchbib2}
\bibliographystyle{utcaps}
\end{document}